\documentclass[11pt]{amsart}

\usepackage{amsmath}
\usepackage{amssymb}
\usepackage{mathrsfs}  
\usepackage{graphicx}
\let\tilde=\widetilde

\newcommand\R{\mathbb{R}}

\newcommand\N{\mathbb{N}}
\newcommand\K{\mathbb{K}}
\newcommand\M{\mathbb{M}}
\newcommand\X{\widehat{X}}
\newcommand\B{\mathcal{B}}
\newcommand{\cF}{\mathcal{F}}
\newcommand\A{\mathcal{F}_{\exp}(\K_0(X))}
\newcommand\F{\mathscr{F}}%{\mathcal{F}C_\mathrm{b}(C_\mathrm{c}(\X),\mathbb{K}(X))}

\renewcommand\i{{\rm 1\kern -.3600em 1}}
\newcommand\n{{(n)}}

\newtheorem{theorem}{Theorem}[section]
\newtheorem{corollary}[theorem]{Corollary}

\newtheorem{proposition}[theorem]{Proposition}
\theoremstyle{definition}
\newtheorem{definition}[theorem]{Definition}
\newtheorem{example}[theorem]{Example}
\newtheorem{remark}[theorem]{Remark}

\numberwithin{equation}{section}
\title[Analysis on the cone of discrete Radon measures ]{Analysis on the cone of discrete Radon measures}

\author[D. Finkelshtein]{Dmitri Finkelshtein}
\address[D. Finkelshtein]{Department of Mathematics, Swansea University, Bay Campus, Swansea SA1 8EN, UK}
\email{{\tt d.l.finkelshtein@swansea.ac.uk}}

\author[Y. Kondratiev]{\fbox{Yuri Kondratiev}}

\author[P. Kuchling]{Peter Kuchling}
\address[P. Kuchling]{Hochschule Bielefeld -  University of Applied Sciences and Arts, 33615 Bielefeld, Germany}
\email{\tt peter.kuchling@hsbi.de}

\author[E. Lytvynov]{Eugene Lytvynov}
\address[E. Lytvynov]{Department of Mathematics, Swansea University, Bay Campus, Swansea SA1 8EN, UK}
\email{{\tt e.lytvynov@swansea.ac.uk}}

\author[M.~J. Oliveira]{Maria Jo\~{a}o Oliveira}
\address[M.~J. Oliveira]{DCeT, Universidade Aberta, P 1269-001 Lisbon, Portugal; CMAFcIO, University of Lisbon, 1749-016 Lisbon, Portugal}
\email{\tt mjoliveira@ciencias.ulisboa.pt}



\begin{document}

\begin{abstract}
We study analysis on the cone of discrete Radon measures over a locally compact Polish space $X$. We discuss probability measures on the cone and the corresponding correlation measures and correlation functions on the sub-cone of finite discrete Radon measures over $X$. For this, we consider on the cone an analogue of the harmonic analysis on the configuration space  developed in \cite{harmana}. We also study elements of finite-difference calculus on the cone: we introduce discrete birth-and-death gradients and study the corresponding Dirichlet forms; finally, we discuss a system of polynomial functions on the cone which satisfy the binomial identity.
\end{abstract}

\maketitle

\begin{center}
  \it In memory of   Prof.~Anatoly Vershik
 \end{center}

%%%%%%%%%%%%%%%%%%%%%%%%%%%%%%%%%%%%%%%%%%%%%%%%%%%%%%%%%%%%%%%%%%%%%%%%%%%%%%%%

\section{Introduction}

Let $X$ be a locally compact Polish space and $\mathcal{B}(X)$ be the corresponding Borel $\sigma$-algebra on $X$. Let $\M(X)$ be the space of real-valued Radon measures on $(X,\mathcal{B}(X))$.

A random measure on $X$ is a random element of the space $\M(X)$. The theory of random measures plays an essential role in the modern studies of probability and point processes, see e.g. \cite{DVJ2, Kal2017}. The (marked) point processes can be interpreted as random \emph{discrete} (non-negative) measures. 
The cone of non-negative discrete Radon measures on $\mathcal{B}(X)$ is the set
\begin{equation}
  \K(X):=\left\{\eta=\sum_i s_i\delta_{x_i}\in\M(X): s_i>0, x_i\in X\right\}, \label{cone_def}
\end{equation}
where $\delta_{x_i}$ denotes the Dirac measure with unit mass at
$x_i\in X$. Here the atoms $x_i$ are assumed to be distinct and their total number is at most countable. By
convention, $\K(X)$ contains the zero measure $\eta=0$, which is represented by the sum over the empty set of
indices $i$. 

The probability measures on the cone $\K(X)$, in particular, the Gamma measure (see Section~2.2 below), are important e.g. for the respresentation theory of big groups \cite{GGV}, study of infinite-dimensional analogue of Lebesgue measure \cite{TVY}, constructions of orthogonal polynomials in the generalisations of the white noise analysis \cite{KL2000,KSSU1998}, etc.

The \emph{support} of a $\eta\in\K(X)$ is defined by
\[
\tau(\eta):=\{x\in X: s_x(\eta):=\eta(\{x\})>0\}.
\]
In particular, $\tau(0)=\emptyset$. If $\eta\in\K(X)$ is fixed, we just  write $s_x:=s_x(\eta)$. Thus, each
$\eta\in\K(X)$ can be rewritten in the form 
\[
\eta=\sum_{x\in\tau(\eta)}s_x\delta_x.
\]
For each $\eta\in\K(X)$ and a compact $\Lambda\in\mathcal{B}(X)$, we have that
\begin{equation*}
  \sum_{x\in\tau(\eta)\cap\Lambda}s_x = \eta(\Lambda)<\infty.
\end{equation*}
Hence, if we interpret $s_x$ as the weight at a point $x\in X$, the total weight in any compact region is finite. Nevertheless, for a typical measure on $\K(X)$,
\begin{equation}\label{infsup}
  |\tau(\eta)\cap\Lambda| = \infty
\end{equation}
for almost all $\eta\in\K(X)$ (henceforth, $|\cdot|$ denotes the cardinality of a discrete set); in other words, the configuration of points constituting the support of $\eta$ is not locally finite. The crucial assumption for \eqref{infsup} is that the measure that provides the distribution of weights $s_x>0$ is infinite on $(0,\infty)$. This makes  analysis on the cone very different from the case of weights (marks) with a finite distribution, the so-called marked configurations with a finite measure on marks, where the results are essentially pretty similar to the case without marks (when all $s_x=1$), see e.g. \cite{KSS1998,
KKS1998, KLU1999}.

The further studies of the analysis and probability on the cone were provided in several directions. The Gibbs perturbation of the Gamma measure was constructed and studied in \cite{MR3041709}, the Laplace operator  and the corresponding diffusion process on the cone were developed in \cite{KLV2015}, see also \cite{CKL2016}, a moment problem on the cone was analysed in \cite{KKL2015}, and some non-equilibrium birth-and-death dynamics on the cone were studied in \cite{FKK2021}.

The present paper deals with two other important topics of the analysis on the cone of discrete Radon measures: we discuss an analogue of the harmonic analysis on the cone, similar to \cite{harmana} and we introduce  elements of  finite-difference calculus on the cone, influenced by \cite{FKLO16}. 

The paper is organised as follows.
In Section~2, we define the basic structures on the cone $\K(X)$ and discuss probability measures on the cone. In particular, in Proposition~\ref{infinbdd},  we describe a class of measures for which \eqref{infsup} holds almost everywhere. In Section~3, we discuss harmonic analysis on the cone. For this, we consider an analogue of the $K$-transform (Definition~\ref{Definition1}), cf.\ \cite{harmana}, which maps functions defined on the sub-cone $\K_0(X)$ of discrete measures with finite support to the cone $\K(X)$, and study properties of the $K$-transform. Next, we discuss the correlation measures and correlation functions of probability measures on the cone. Finally, in Section~4, we study discrete gradients on the cone, derive several properties of the corresponding Dirichlet forms (Propositions~\ref{propZh1} and \ref{propZh2}), and describe properties of certain polynomial functions on the cone which satisfy, in particular, the binomial formula (Proposition~\ref{prop:binom:cone}).

\medskip

\hrule

\medskip

This article was written with the enthusiasm, commitment and deep scientific knowledge of our collaborator, co-author
and friend Yuri Kondratiev. He passed away on September 5\textsuperscript{th}, 2023.

\section{Framework}
Let $\M(X)$  denote the space of real-valued Radon measures on $X$. We equip $\M(X)$ with the vague topology,
that is the coarsest topology on $\M(X)$ 
such that, for each continuous function $f:X\to\R$ with compact support (shortly $f\in C_\mathrm{c}(X)$), the following
mapping is continuous:
\[
\M(X)\ni\nu\mapsto\langle\nu,f\rangle:=\int_Xf(x)\,d\nu(x)\in\R.
\]
We consider the corresponding Borel $\sigma$-algebra $\B(\M(X))$. Let $\K(X)\subset \M(X)$ be defined by \eqref{cone_def}.
We endow $\K(X)$ with the vague topology induced by $\M(X)$ and
denote the corresponding Borel $\sigma$-algebra by $\mathcal{B}(\K(X))$; that is then the trace $\sigma$-algebra of $\B(\M(X))$.  Note that, for each $f\in C_\mathrm{c}(X)$, 
\begin{equation*}%\label{pairing}
  \langle\eta,f\rangle=\sum_{x\in\tau(\eta)}s_xf(x), \qquad \eta\in\K(X). 
\end{equation*}

Let $\R^*_+:=(0,+\infty)$ be endowed with the logarithmic metric
$$d_{\R^*_+}(s_1,s_2):=\left|\ln(s_1)-\ln(s_2)\right|,\quad s_1,s_2>0.$$ It is easy to see that $\R^*_+$ is a locally compact Polish space, and any set of the form $\left[a,b\right]$, with $0<a<b<\infty$, is compact. Then, $\X:=\R^*_+\times X$ is also a locally compact Polish space. Hence, we can define $\M(\X)$ and $\B_\mathrm{c}(\X)$ as before. Clearly, any compact set in $\X$ is a subset of $[a,b]\times\Lambda$ for some $0<a<b<\infty$ and $\Lambda\in\B_\mathrm{c}(X)$; henceforth, $\mathcal{B}_\mathrm{c}(X)$ denotes the family of sets in $\mathcal{B}(X)$ with compact closure.

We define the space of \emph{locally finite configurations} over $\X$ as the set
\begin{equation*}
\Gamma(\X):=\Bigl\{\gamma\subset \X \colon \bigl\lvert \gamma\cap ([a,b]\times\Lambda)\bigr\rvert<\infty\ \  \forall\ \Lambda\in\B_\mathrm{c}(X), b>a>0\Bigr\}.
%\label{Eq1}
\end{equation*}
As usual, each configuration $\gamma\in\Gamma(\X)$ can be identified with the Radon measure
$\sum_{y\in\gamma}\delta_y\in\M(\X)$. Thus, the inclusion $\Gamma(\X)\subset\M(\X)$ holds, which allows to endow $\Gamma(\X)$
with the vague topology induced by $\M(\X)$ and the corresponding Borel $\sigma$-algebra $\mathcal{B}(\Gamma(\X))$.

Let $\Gamma_p(\X)\subset\Gamma(\X)$ be the set of \emph{pinpointing configurations}, which consists of 
all configurations $\gamma$ such that if $(s_1,x_1), (s_2,x_2)\in\gamma$ with $x_1=x_2$, then $s_1=s_2$. For each 
$\gamma\in\Gamma_p(\X)$ and each $\Lambda\in\mathcal{B}_\mathrm{c}(X)$, we define the local mass of $\Lambda$ by
\[
\gamma(\Lambda):=\int_{\X} s\i_\Lambda(x)\,d\gamma(s,x)=
\sum_{(s,x)\in\gamma}s\i_\Lambda(x)\in\left[0,+\infty\right],
\]
where $\i_\Lambda$ denotes the indicator function of the set $\Lambda$. The 
set of pinpointing configurations with finite local mass is then defined by
\[
\Pi(\X):=\{\gamma\in\Gamma_p(\X): \gamma(\Lambda)<\infty\ \forall\Lambda\in\mathcal{B}_\mathrm{c}(X)\}.
\]
One has $\Pi(\X)\in\mathcal{B}(\Gamma(\X))$ and we fix the trace $\sigma$-algebra
of $\mathcal{B}(\Gamma(\X))$ on $\Pi(\X)$, denoted by $\mathcal{B}(\Pi(\X))$.

Thus, we have defined a bijective mapping
\begin{equation}
  \begin{array}
    [c]{rl}
    \mathcal{R}:\Pi(\X)&\longrightarrow\K(X)\\
    \gamma=\displaystyle\sum_{(s,x)\in\gamma}\delta_{(s,x)}&\longmapsto\displaystyle\sum_{(s,x)\in\gamma}s\delta_x.
    \end{array}\label{def:cR}
\end{equation}
As shown in \cite{MR3041709}, both $\mathcal{R}$ and its inverse mapping $\mathcal{R}^{-1}$ are
measurable with respect to $\mathcal{B}(\Pi(\X))$ and $\mathcal{B}(\K(X))$ and $\mathcal{R}$ is a
$\sigma$-isomorphism. Moreover,  
\[
\mathcal{B}(\K(X))=\{\mathcal{R}(A): A\in\B(\Pi(\X))\},
\]
and thus $\mathcal{B}(\Pi(\X))=\{\mathcal{R}^{-1}(B):B\in\mathcal{B}(\K(X))\}$.

Let $\pi_{\nu\otimes\sigma}$ be the \emph{Poisson measure} on $\mathcal{B}(\Gamma(\X))$  (see e.g. \cite{AKR1998a} for the details) with intensity measure
$\nu\otimes\sigma$, where $\nu$ and $\sigma$ are two non-atomic positive Radon measures on the Borel
$\sigma$-algebras $\mathcal{B}(\R^*_+)$ and $\mathcal{B}(X)$, respectively, and $\nu$ has finite first moment, i.e.,
\begin{equation}
  \int_{\R^*_+}s\,d\nu(s)<\infty.\label{fmoment}
\end{equation}
In terms of the Laplace transform of $\pi_{\nu\otimes\sigma}$, for each continuous function
$f:\X\to[0,\infty)$ with compact support, we have 
\[
\int_{\Gamma(\X)} e^{-\langle\gamma,f\rangle}d\pi_{\nu\otimes\sigma}(\gamma)=\exp\left(\int_{\X}\left(e^{-f(s,x)}-1\right)d(\nu\otimes\sigma)(s,x)\right).
\]

\begin{proposition}\label{support}
For $\nu$ and $\sigma$ as above,  $\pi_{\nu\otimes\sigma}(\Pi(\X))=1$.
\end{proposition}

\begin{proof}
As easily seen, $\Gamma_p(\X),\Pi(\X)\in\mathcal{B}(\Gamma(\X))$. Moreover, using the distribution of
configurations of the form $\gamma\cap(\left[a,b\right]\times\Lambda)$, $\gamma\in\Gamma(\X)$, $0<a<b$,
$\Lambda\in\mathcal{B}_\mathrm{c}(X)$, under $\pi_{\nu\otimes\sigma}$ (see e.g.~\cite{K93}), we conclude that, for each
$0<a<b$, $\Lambda\in\mathcal{B}_\mathrm{c}(X)$ fixed,
\begin{align*}
&\pi_{\nu\otimes\sigma}\!\left(\!\{\gamma\in\Gamma(\X): \exists\,(s_1,x_1),(s_2,x_2)\in\gamma\cap(\left[a,b\right]\times\Lambda)\,\mathrm{s.t.~} x_1=x_2, s_1\not=s_2\}\!\right)\\
&=0.
\end{align*}
Thus $\pi_{\nu\otimes\sigma}(\Gamma_p(\X))=1$. The rest of the proof is a consequence of the Mecke identity
\cite{Mec67}, which states that for every measurable function
$H:\Gamma(\X)\times \X\to[0,+\infty)$ the following equality holds:
\begin{align}
&\int_{\Gamma(\X)}\int_{\X}H(\gamma,s,x)\,d\gamma(s,x)d\pi_{\nu\otimes\sigma}(\gamma)\label{MeckeId}\\
&=\int_{\Gamma(\X)}\int_{\X}H(\gamma\cup\{(s,x)\},s,x)\,d(\nu\otimes\sigma)(s,x)d\pi_{\nu\otimes\sigma}(\gamma)\nonumber.
\end{align}
In particular, for each $\Lambda\in\mathcal{B}_\mathrm{c}(X)$ fixed and for the measurable function
$H(\gamma,s,x)=s\i_\Lambda(x)$, the latter leads to
\begin{align}
\int_{\Gamma_p(\X)}\gamma(\Lambda)d\pi_{\nu\otimes\sigma}(\gamma)
&=\int_{\Gamma_p(\X)}\int_{\X}s\i_\Lambda(x)\,d(\nu\otimes\sigma)(s,x)d\pi_{\nu\otimes\sigma}(\gamma)\notag\\
&=\sigma(\Lambda)\int_{\R^*_+}s\,d\nu(s)<\infty,\label{finite1stmoment}
\end{align}
which implies that $\gamma(\Lambda)<\infty$ for $\pi_{\nu\otimes\sigma}$-a.a.~$\gamma\in\Gamma_p(\X)$.
Hence, $\pi_{\nu\otimes\sigma}(\Pi(\X))=1$.
\end{proof}

\begin{definition}
  By Proposition~\ref{support}, we may view $\pi_{\nu\otimes\sigma}$ as a probability
measure on $\mathcal{B}(\Pi(\X))$. Thus, we consider the push-forward of the measure $\pi_{\nu\otimes\sigma}$
under $\mathcal{R}$ to $\K(X)$, which we denote by $\pi_{\K,\nu\otimes\sigma}$.
\end{definition}

Then, by \eqref{finite1stmoment},
\begin{equation*}\label{1stlocmom}
  \int_{\K(X)} \eta(\Lambda) d\pi_{\K,\nu\otimes\sigma}(\eta)<\infty.
\end{equation*}

Note that, for each continuous bounded function $g\in C_\mathrm{b}(\R^N)$, $N\in\N$, and for every functions $\varphi_1,\dotsc,\varphi_N\in C_\mathrm{c}(X)$, we have
\begin{align*}
\int_{\K(X)} &g\left(\langle\eta,\varphi_1\rangle,\dotsc,\langle\eta,\varphi_N\rangle\right)d\pi_{\K,\nu\otimes\sigma}(\eta)\\
&=\int_{\Pi(\X)}g\left(\langle\gamma,\mathrm{id}\otimes\varphi_1\rangle,\dotsc,\langle\gamma,\mathrm{id}\otimes\varphi_N\rangle\right)d\pi_{\nu\otimes\sigma}(\gamma),
\end{align*}
where $(\mathrm{id}\otimes\varphi)(s,x):=s\varphi(x)$. In terms of the Laplace transform of
$\pi_{\K,\nu\otimes\sigma}$, for each $f\in C_\mathrm{c}(X)$ one finds
\[
\int_{\K(X)} e^{-\langle\eta,f\rangle}d\pi_{\K,\nu\otimes\sigma}(\eta)=\exp\left(\int_X\int_{\R^*_+}\left(e^{-sf(x)}-1\right)d\nu(s)d\sigma(x)\right).
\]

\begin{proposition}[Mecke-type identity]
For each measurable function $F:\K(X)\times\X\to[0,+\infty)$, the following equality holds
\begin{multline}\label{yfuftyttfyoi9u}
\int_{\K(X)}\int_XF(\eta,s_x,x)\,d\eta(x)d\pi_{\K,\nu\otimes\sigma}(\eta)
\\=\int_{\K(X)}\int_X\int_{\R^*_+}s\,F(\eta+s\delta_x,s,x)\,d\nu(s)d\sigma(x)d\pi_{\K,\nu\otimes\sigma}(\eta).
\end{multline}
\end{proposition}

\begin{proof} 
Formula \eqref{yfuftyttfyoi9u} is a direct consequence of the Mecke
identity \eqref{MeckeId} and Proposition~\ref{support}:
\begin{align*}
&\int_{\K(X)}\int_XF(\eta,s_x,x)\,d\eta(x)d\pi_{\K,\nu\otimes\sigma}(\eta)\\
&=\int_{\Pi(\X)}\int_{\X}sF(\mathcal{R}(\gamma),s,x)\,d\gamma(s,x)d\pi_{\nu\otimes\sigma}(\gamma)\\
&=\int_{\Pi(\X)}\int_{\X}sF(\mathcal{R}(\gamma\cup\{(s,x)\}),s,x)\,d(\nu\otimes\sigma)(s,x)d\pi_{\nu\otimes\sigma}(\gamma)\\  
&=\int_{\K(X)}\int_X\int_{\R^*_+}sF(\eta+s\delta_x,s,x)\,d\nu(s)d\sigma(x)d\pi_{\K,\nu\otimes\sigma}(\eta).\qedhere
\end{align*}
\end{proof}

A special case concerns $\nu=\nu_\theta$, where
\[
\nu_\theta(ds):=\frac{\theta}{s}e^{-s}ds
\]
for some $\theta>0$. In this case, the corresponding Poisson measure is called a Gamma--Poisson measure and the
push-forward to $\K(X)$ is called a \emph{Gamma measure} with intensity $\theta$ and denoted by
$\mathcal{G}_\theta$. Alternatively, a Gamma measure can be characterized by its Laplace transform \cite{CCR}: for
each $-1<f\in C_\mathrm{c}(X)$,
\[
\int_{\K(X)} e^{-\langle\eta,f\rangle}d\mathcal{G}_\theta(\eta)=\exp\left(-\theta\int_X\ln(1+f(x))\, d\sigma(x)\right).
\]

The next result states a Mecke-type characterization result for Gamma measures. There, $\M_+(X)\subset\M(X)$
denotes the cone of all non-negative measures.

\begin{proposition}[\cite{MR3041709}] 
Let $\mu$ be a probability measure defined on $\M_+(X)$ which has finite first local moments, i.e., for each
$\Lambda\in\mathcal{B}_\mathrm{c}(X)$,
\[
\int_{\M_+(X)}\eta(\Lambda)\,d\mu(\eta)<\infty.
\]
Then, $\mu=\mathcal{G}_\theta$ if and only if for any measurable function
$G:\M_+(X)\times X\to[0,+\infty)$ we have
\begin{multline*}
\int_{\M_+(X)}\int_XG(\eta,x)\,d\eta(x)d\mu(\eta)\\=\int_{\M_+(X)}\int_X\int_{\R^*_+}sG(\eta+s\delta_x,x)\,d\nu_\theta(s)d\sigma(x)d\mu(\eta).
\end{multline*}
\end{proposition}

Measure $\nu_\theta$ is an example of an infinite measure on $\R^*_+$. The next statement shows that, for any such $\nu$ (with finite first moment), the support $\tau(\eta)$ of $\pi_{\K,\nu\otimes\sigma}$-a.a.~$\eta\in\K(X)$ is not a locally finite subset of $X$.
\begin{proposition}\label{infinbdd}
  Let $\nu(\R^*_+)=\infty$. Then, for any $\Lambda\in\B_\mathrm{c}(X)$ with $\sigma(\Lambda)>0$ and for $\pi_{\K,\nu\otimes\sigma}$-a.a.~$\eta\in\K(X)$, the set $\tau(\eta)\cap\Lambda$ has an infinite number of elements. 
\end{proposition}
\begin{proof}
By e.g. \cite{AKR1998a}, for any $\Lambda\in\B_\mathrm{c}(X)$, $0<a<b$, $n\in\N_0:=\N\cup\{0\}$,
\begin{equation}\label{npointset}
  \pi_{\nu\otimes\sigma}\Bigl\{\gamma\in\Gamma(\X) \colon \bigl\lvert \gamma\cap([a,b]\times \Lambda)\bigr\rvert = n\Bigr\}=\frac{\bigl(\sigma(\Lambda)\nu([a,b])\bigr)^n}{n!}e^{-\sigma(\Lambda)\nu([a,b])}.
\end{equation}
Therefore, if $\sigma(\Lambda)>0$ and $\nu(\R^*_+)=\infty$ then, passing $a\to0$, $b\to\infty$, we obtain, by the definition of $\pi_{\K,\nu\otimes\sigma}$, that
\[
  \pi_{\K,\nu\otimes\sigma}(\{\eta\in\K(X): |\tau(\eta)\cap\Lambda|=n\})=0,
\]
that implies the statement.
\end{proof}

\begin{remark}
  Note that by \eqref{npointset}, for any $\Lambda\in\B_\mathrm{c}(X)$ with $\sigma(\Lambda)=0$,
  \[
    \pi_{\K,\nu\otimes\sigma}(\{\eta\in\K(X): \tau(\eta)\cap\Lambda=\emptyset\})=1.
  \]
\end{remark}

\begin{definition}
Let  $\mathcal{M}^1_{\mathrm{fm}}(\K(X))$ denote the space of all   
probability measures $\mu$ on $(\K(X),\B(\K(X))$ which have finite local moments of all orders, that is, 
for all $n\in\N$ and all $\Lambda\in\B_\mathrm{c}(X)$,
\begin{equation*}
  \int_{\K(X)}\left(\eta(\Lambda)\right)^n\,d\mu(\eta)<\infty.
\end{equation*}
\end{definition}

\begin{proposition}\label{Proposition5}
$\pi_{\K,\nu\otimes\sigma}\in\mathcal{M}^1_{\mathrm{fm}}(\K(X))$ if and only if $\nu$ has finite moments of all orders, that
is, for all $n\in\N$,
\begin{equation}
  \int_{\R^*_+}s^n\,d\nu(s)<\infty.\label{allmoments}
\end{equation}
\end{proposition}

\begin{proof}
For each $n\in\N$ and each $\Lambda\in\B_\mathrm{c}(X)$ we have
\begin{align*}
\int_{\K(X)}\left(\eta(\Lambda)\right)^n\,d\pi_{\K,\nu\otimes\sigma}(\eta)
&=\int_{\Pi(\X)}\left(\gamma(\Lambda)\right)^n\,d\pi_{\nu\otimes\sigma}(\gamma)\\
&=\int_{\Gamma(\X)}\langle\gamma,\mathrm{id}\otimes\i_\Lambda\rangle^n\,d\pi_{\nu\otimes\sigma}(\gamma),
\end{align*}
where, as shown e.g. in \cite[Theorem 6.1]{FKLO21}, the latter integral is equal to
\[
\sum_{k=1}^n\frac{(\sigma(\Lambda))^k}{k!}\sum_{\substack{(i_1,\ldots,i_k)\in\N^k\\i_1+\ldots+i_k=n}}\binom{n}{i_1\ldots i_k}
\prod_{j=1}^k\left(\int_{\R^*_+}s^{i_j}\,d\nu(s)\right).
\]
The required necessary and sufficient condition then follows.
\end{proof}

\begin{remark}
By Proposition~\ref{Proposition5}, $\mathcal{G}_\theta\in\mathcal{M}^1_{\mathrm{fm}}(\K(X))$ for every intensity $\theta>0$.
\end{remark}

  Let $\Lambda\in\B_\mathrm{c}(X)$ and $0<a<b$ be fixed. 
  We consider 
  \[
    \Gamma([a,b]\times\Lambda):=\{\gamma\in\Gamma(\X) \colon \gamma\subset [a,b]\times\Lambda\} 
  \] 
  and let $\B_{[a,b]\times\Lambda}(\Gamma(\X))$ be the corresponding $\sigma$-algebra on $\Gamma(\X)$ which is $\sigma$-isomorphic to the trace $\sigma$-algebra 
  $\B(\Gamma([a,b]\times\Lambda))$, 
  see \cite[Section 2.2]{harmana} for details. Let $\Pi([a,b]\times\Lambda)=\Gamma_p([a,b]\times\Lambda)$ be the corresponding measurable subset of pinpointing configurations (which are finite and, hence, have finite local mass), and let $\B_{[a,b]\times\Lambda}(\Pi(\X))$ be the corresponding $\sigma$-algebra on $\Pi(\X)$. Then
  \begin{align*}
    \K([a,b]\times\Lambda)&:=\mathcal{R}(\Pi([a,b]\times\Lambda))\\&=\{\eta\in\K(X)\colon \tau(\eta)\subset\Lambda,\ s_x(\eta)\in[a,b]\ \forall x\in\tau(\eta)\}\in \B(\K(X)),
  \end{align*}
  and we can consider the corresponding $\sigma$-algebra $\B_{[a,b]\times\Lambda}(\K(X))$. 
  We also consider the measurable mapping
  $\hat{p}_{[a,b]\times\Lambda}:\Pi(\X)\to \Pi([a,b]\times\Lambda)$ given by 
  $\hat{p}_{[a,b]\times\Lambda}(\gamma):=\gamma\cap([a,b]\times\Lambda)$. We then define the measurable mapping $p_{\Lambda,a,b}:\K(X)\to \K([a,b]\times\Lambda)$ by
  \begin{equation*}
      p_{\Lambda,a,b}(\eta) := \mathcal{R} \,\hat{p}_{[a,b]\times\Lambda}\, \mathcal{R}^{-1}\eta, \quad\eta\in\K(X).
  \end{equation*}
  Then, for each $\eta\in\K(X)$,
    \[
       p_{\Lambda,a,b}(\eta) =\sum_{x\in\tau(\eta)\cap\Lambda}\i_{[a,b]}(s_x)s_x\delta_x,
    \]
  and
  \begin{equation}
    \lvert \tau(  p_{\Lambda,a,b}(\eta) ) \rvert  = \sum_{x\in\tau(\eta)\cap\Lambda}\i_{[a,b]}(s_x)\leq
  \frac{1}{a}\sum_{x\in\tau(\eta)\cap\Lambda}s_x = \frac{1}{a}\eta(\Lambda)<\infty.\label{eq00}
  \end{equation}

\begin{definition}
  A measure $\mu\in\mathcal{M}^1_{\mathrm{fm}}(\K(X))$ is called locally absolutely continuous with respect to $\pi_{\K,\nu\otimes\sigma}$ if, for all $\Lambda\in\B_\mathrm{c}(X)$ and $0<a<b$, the measure $\mu^{\Lambda,a,b}:=\mu\circ   p_{\Lambda,a,b}^{-1}$ is absolutely continuous with respect to  
  $\pi_{\K,\nu\otimes\sigma}^{\Lambda,a,b}:=\pi_{\K,\nu\otimes\sigma}\circ   p_{\Lambda,a,b}^{-1}$. Equivalently, the push-forward of the measure $\mu$ under $\mathcal{R}^{-1}$ to $\Pi(\X)$ is locally absolutely continuous with respect to $\pi_{\nu\otimes\sigma}$ (see \cite{harmana} for details).
\end{definition}

\section{Harmonic analysis on the cone}

\subsection{Discrete Radon measures with finite support}

We consider the sub-cone of all non-negative Radon measures with finite support
\[
\K_0(X):=\left\{\eta\in\K(X): |\tau(\eta)|<\infty\right\}=\bigsqcup_{n=0}^\infty\K_0^{(n)}(X),
\]
where $\K_0^{(0)}(X):=\{0\}$ is the set consisting of the zero measure and, for each $n\in\N$,
$\K_0^{(n)}(X):=\{\eta\in\K_0(X): |\tau(\eta)|=n\}$. 

We consider also the space $\Pi_0(\X):=\{\gamma\in\Pi(\X): |\gamma|<\infty\}$ of pinpointing finite configurations
on $\X$. Since $\Pi_0(\X)\subset\Gamma_0(\X):=\{\gamma\in\Gamma(\X): |\gamma|<\infty\}$, we can endow $\Pi_0(\X)$
with the topology induced by the topology defined on $\Gamma_0(\X)$ \cite{harmana} and the corresponding
Borel $\sigma$-algebra, denoted by $\B(\Pi_0(\X))$. As easily seen, $\Pi_0(\X)\in\mathcal{B}(\Gamma_0(\X))$, thus, $\B(\Pi_0(\X))=\mathcal{B}(\Gamma_0(\X))\cap \Pi_0(\X)$. Also, $\K_0(X)=\mathcal{R}(\Pi_0(\X))$, which allows to endow $\K_0(X)$ with the $\sigma$-algebra
\[
\B(\K_0(X)):=\{\mathcal{R}(A): A\in\B(\Pi_0(\X))\}
\]
with respect to which the restriction $\mathcal{R}_{|\Pi_0(\X)}:\Pi_0(\X)\to\K_0(X)$ is a $\sigma$-isomorphism.

A set $A\in \B(\K_0(X))$ is called \emph{bounded} if there exist $0<a<b$, $\Lambda\in\B_\mathrm{c}(X)$ and an $N\in\N$
such that, for all $\eta\in A$,
\begin{equation}
  \tau(\eta)\subset\Lambda,\qquad |\tau(\eta)|\leq N,\qquad s_x\in\left[a,b\right] \ \ \forall\,x\in\tau(\eta).\label{Eq4}
\end{equation}
We denote by $\B_\mathrm{b}(\K_0(X))$ the set of all bounded sets in $\B(\K_0(X))$.

Let $L^0(\K_0(X))$ be the class of all $\B(\K_0(X))$-measurable functions $G:\K_0(X)\to\R$. Assume that $G=\i_A G$, where a $A\in\B(\K_0(X))$ is such that, for some $0<a<b$ and $\Lambda\in\B_\mathrm{c}(X)$, we have, for all for all $\eta\in A$,
\begin{equation}
  \tau(\eta)\subset\Lambda, \qquad s_x\in\left[a,b\right] \ \ \forall\,x\in\tau(\eta).\label{Eq41}
\end{equation}
 Then we will say that $G$ has a \emph{local support} on $\K_0(X)$. The class of all measurable functions with local support is denoted by $L^0_{\mathrm{ls}}(\K_0(X))$. Similarly, if $G=\i_A G$ for a bounded $A\in\B_\mathrm{b}(\K_0(X))$, we will say that $G$ has \emph{bounded support} on $\K_0(X)$. We consider a subclass of all \emph{bounded} measurable functions with bounded support,  denoted by $B_{\mathrm{bs}}(\K_0(X))$. Thus, $B_{\mathrm{bs}}(\K_0(X))\subset L^0_{\mathrm{ls}}(\K_0(X))\subset L^0(\K_0(X))$. We will denote the pre-images of these classes under $(\mathcal{R}_{|\Pi_0(\X)})^{-1}$ by $B_{\mathrm{bs}}(\Pi_0(\X))$, $L^0_{\mathrm{ls}}(\Pi_0(\X))$, $L^0(\Pi_0(\X))$, respectively.

\subsection{$K$-transform}

In the sequel, for $\eta$, $\xi\in\K(X)$, we write $\xi\subset\eta$ if $\tau(\xi)\subset\tau(\eta)$ and
$s_x(\xi)=s_x(\eta)$ for all $x\in\tau(\xi)$. If, additionally, $\xi\in\K_0(X)$, we write $\xi\Subset\eta$. Note that, for $\xi,\eta\in\K(X)\subset\M(X)$ with $\xi\subset\eta$, the measure $\eta-\xi\in\M(X)$ is well-defined and $\eta-\xi\in\K(X)$ with $\tau(\eta-\xi)=\tau(\eta)\setminus\tau(\xi)$ and $s_x(\eta-\xi)=s_x(\eta)$ for each $x\in \tau(\eta)\setminus\tau(\xi)$.

\begin{definition}\label{Definition1}
  For each $G\in L^0_{\mathrm{ls}}(\K_0(X))$, we define a function $KG:\K(X)\to\R$, called the $K$-transform of $G$, by  
  \begin{equation}
  (KG)(\eta):=\sum_{\xi\Subset\eta}G(\xi),\quad\eta\in\K(X).\label{Eq3}
  \end{equation}
  \end{definition}
Note that since $G\in L^0_{\mathrm{ls}}(\K_0(X))$, there exist $0<a<b$ and $\Lambda\in\B_\mathrm{c}(X)$ such that
\[
  \sum_{\xi\Subset\eta}G(\xi)=\sum_{\xi\Subset\eta}\i_{\substack{\tau(\xi)\subset\Lambda\\s_x\in[a,b], x\in\tau(\xi)}}(\xi) G(\xi) = \sum_{\xi\subset  p_{\Lambda,a,b}(\eta) } G(\xi),
\]
and hence, by \eqref{eq00}, the sum in \eqref{Eq3} is well-defined (note that, by \eqref{eq00}, $  p_{\Lambda,a,b}(\eta) \in\K_0(X)$). Moreover, we have shown that, for $F:=KG$, $G\in L^0_{\mathrm{ls}}(\K_0(X))$, 
\begin{equation}
F(\eta)=F(  p_{\Lambda,a,b}(\eta) ),\quad\eta\in\K(X), \label{5Eq1}
\end{equation}
for some $0<a<b$ and $\Lambda\in\B_\mathrm{c}(X)$ (dependent on $F$).

\begin{proposition}\label{KPi}
Let $G\in L^0_{\mathrm{ls}}(\K_0(X))$ and $KG$ be defined by \eqref{Eq3}. The $KG$ is a $\B(\K(X))$-measurable function. 
\end{proposition}  

\begin{proof}
 Let $F\in L^0_{\mathrm{ls}}(\Pi_0(\X))$. We define the $\mathcal{B}(\Gamma_0(\X))$-measurable function $\tilde F:\Gamma_0(\X)\to\R$ given by
 \[
  \tilde F(\eta):=\begin{cases}
  F(\eta),&\quad \text{if } \eta\in\Pi_0(\X),\\[1mm]
  0, & \quad \text{otherwise}. 
  \end{cases}
\]
By \cite{harmana}, the function $K_{\X}\tilde F:\Gamma(\X)\to\R$, given by
\[
  (K_{\X}\tilde F)(\gamma):=\sum_{\substack{\zeta\subset\gamma\\ |\zeta|<\infty}}\tilde F(\zeta),\quad\gamma\in\Gamma(\X),
\]
is well-defined and $\B(\Gamma(\X))$-measurable. Therefore, the restriction 
\[
  (K_\Pi F)(\gamma):=\Bigl((K_{\X}\tilde F)\bigr\rvert_{\Pi(\X)}\Bigr)(\gamma)
  = \sum_{\substack{\zeta\subset\gamma\\ |\zeta|<\infty}}F(\zeta),\quad\gamma\in\Pi(\X)
\]
is $\B(\Pi(\X))$-measurable.

Now, let $G\in L^0_{\mathrm{ls}}(\K_0(X))$ be given and consider $F = G\circ\mathcal{R}\bigr\rvert_{\Pi_0(\X)}\in L^0_{\mathrm{ls}}(\Pi_0(\X))$.
Note that
\begin{equation}
KG=\sum_{n=0}^\infty K(G\i_{\K^{(n)}_0(X)}).\label{Eq2}
\end{equation}
For an arbitrary set of indices $I\subseteq\N_0$, let $\eta=\sum_{i\in I}s_i\delta_{x_i}\in\K(X)$. 
Hence, by \eqref{Eq2},
\begin{align}\label{eqd2qwr2413}
(KG)(\eta)&=\sum_{n=0}^\infty\sum_{\{i_1,\ldots,i_n\}\subseteq I}G\left(\sum_{i=1}^ns_{i_k}\delta_{x_{i_k}}\right)\\&
=\sum_{n=0}^\infty\sum_{\{i_1,\ldots,i_n\}\subseteq I}G\left(\mathcal{R}\left(\sum_{i=1}^n\delta_{(s_{i_k}, x_{i_k})}\right)\right)\notag\\
&=K_\Pi(G\circ\mathcal{R}_{|\Pi_0(\X)})(\mathcal{R}^{-1}\eta)=
(K_\Pi F)(\mathcal{R}^{-1}\eta)\notag
\end{align}
with $\mathcal{R}^{-1}\eta=\sum_{i\in I}\delta_{(s_i,x_i)}\in\Pi(\X)$. 

Since $K_\Pi F$ is $\B(\Pi(\X))$-measurable, we have that $KG$ is $\B(\K(X))$-meas\-ur\-able.
\end{proof}

Let $\mathcal{F}L^0(\K(X))$ be the class of all $\B(\K(X))$-measurable functions $F:\K(X)\to\R$ such that \eqref{5Eq1} holds (for some $0<a<b$ and $\Lambda\in\B_\mathrm{c}(X)$ dependent on $F$). By~\eqref{5Eq1} and Proposition~\ref{KPi}, $K:L^0_{\mathrm{ls}}(\K_0(X))\to \mathcal{F}L^0(\K(X))$.

\begin{proposition}
  
1.~The mapping $K:L^0_{\mathrm{ls}}(\K_0(X))\to \mathcal{F}L^0(\K(X))$ is linear, positivity-preserving and invertible, with the inverse mapping given by, for $F\in \mathcal{F}L^0(\K(X))$,
  \begin{equation}
    (K^{-1}F)(\eta)=\sum_{\xi\subset\eta}(-1)^{|\tau(\eta)|-|\tau(\xi)|}F(\xi),\qquad \eta\in\K_0(X).\label{inverse}
  \end{equation}
2.~For each $G\in B_{\mathrm{bs}}(\K_0(X))$, there exist $C>0$, $\Lambda\in\B_\mathrm{c}(X)$ and $N\in\N$ such that
\begin{equation*}
  \left|(KG)(\eta)\right|\leq C (1+\eta(\Lambda))^N,\quad\eta\in\K(X). %\label{polbddp}
\end{equation*}
\end{proposition}  
  
  \begin{proof}
  1.~The linearity and positivity-preserving properties follow directly from the definition \eqref{Eq3}. Next, for $G\in L^0_{\mathrm{ls}}(\K_0(X))$, we set $F=KG$ and denote the right-hand side of \eqref{inverse} by $K^{-1}{F}$. Then
  \begin{align*}
    (K^{-1}F)(\eta)&=\sum_{\xi\subset\eta}(-1)^{|\tau(\eta)|-|\tau(\xi)|}\sum_{\zeta\subset\xi} G(\zeta)=
    \sum_{\zeta\subset\eta} G(\zeta)\sum_{\substack{\xi\subset\eta:\\\zeta\subset\xi}}(-1)^{|\tau(\eta)|-|\tau(\xi)|}\\
    &=\sum_{\zeta\subset\eta} G(\zeta)\sum_{\xi\subset\eta-\zeta}(-1)^{|\tau(\eta-\zeta)|-|\tau(\xi)|}=\sum_{\zeta\subset\eta} G(\zeta)0^{|\tau(\eta-\zeta)|}=G(\eta).
  \end{align*}
  On the other hand, let $F\in \mathcal{F}L^0(\K(X))$ and let $\Lambda\in\B_\mathrm{c}(X)$ and $0<a<b$ be such that \eqref{5Eq1} holds.
  We have 
  \begin{align*}
    &\quad (K^{-1}F)(\eta)=\sum_{\xi\subset\eta}(-1)^{|\tau(\eta)|-|\tau(\xi)|}F(\xi)\\&=\sum_{\substack{\xi_1\subset\eta:\\
    \tau(\xi_1)\subset \Lambda,\\
    s_x\in[a,b] \ \forall x\in  \tau(\xi_1)}}
    \sum_{\substack{\xi_2\subset\eta:\\
    \exists x\in  \tau(\xi_2):\,  x\in X\setminus \Lambda \text{ or } s_x\notin[a,b]}}
    (-1)^{|\tau(\eta)|-|\tau(\xi_1)|-|\tau(\xi_2)|}F(\xi_1+\xi_2)
    \\&=\sum_{\substack{\xi_1\subset\eta:\\
    \tau(\xi_1)\subset \Lambda,\\
    s_x\in[a,b] \ \forall x\in  \tau(\xi_1)}}F(\xi_1)(-1)^{|\tau(\eta)|-|\tau(\xi_1)|}
    \sum_{\substack{\xi_2\subset\eta:\\
    \exists x\in  \tau(\xi_2):\,  x\in X\setminus \Lambda \text{ or } s_x\notin[a,b]}}
    (-1)^{-|\tau(\xi_2)|};
  \end{align*}
  and since
  \[
    \sum_{\substack{\xi_2\subset\eta:\\
    \exists x\in  \tau(\xi_2):\,  x\in X\setminus \Lambda \text{ or } s_x\notin[a,b]}}
    (-1)^{-|\tau(\xi_2)|} = 0^{\lvert \{x\in\tau(\eta)\colon x\in X\setminus\Lambda \text{ or } s_x\notin[a,b]\}\rvert},
  \]
  we obtain
  \[
    (K^{-1}F)(\eta) =\i_{\tau(\eta)\subset\Lambda,\, s_x\in[a,b] \ \forall x\in\tau(\eta)} (\eta) (K^{-1}F)(\eta),
  \]
  thus $K^{-1}F\in L^0_{\mathrm{ls}}(\K_0(X))$.  We then have 
  \begin{align*}
    (KK^{-1}F)(\eta)&= \sum_{\xi\Subset\eta} \i_{\tau(\xi)\subset\Lambda,\, s_x\in[a,b] \ \forall x\in\tau(\xi)} (\xi) \sum_{\zeta\subset\xi}(-1)^{|\tau(\xi)|-|\tau(\zeta)|}F(\zeta)\\
    &= \sum_{\xi\subset   p_{\Lambda,a,b}(\eta) }  
     \sum_{\zeta\subset\xi}(-1)^{|\tau(\xi)|-|\tau(\zeta)|}F(\zeta)\\
    &= \sum_{\zeta\subset  p_{\Lambda,a,b}(\eta) } F(\zeta) \sum_{\substack{\xi\subset   p_{\Lambda,a,b}(\eta) :\\\zeta\subset\xi}}  
    (-1)^{|\tau(\xi)|-|\tau(\zeta)|}\\
    &= \sum_{\zeta\subset  p_{\Lambda,a,b}(\eta) } F(\zeta) 0^{|\tau(  p_{\Lambda,a,b}(\eta) )|-|\tau(\zeta)|}= F(  p_{\Lambda,a,b}(\eta) )=F(\eta),
  \end{align*}
  by \eqref{5Eq1}.
  
2.~For $G\in B_{\mathrm{bs}}(\K_0(X))$, we have $|G|\leq c\i_A$ for some $c>0$, $A\in\B_\mathrm{b}(\K_0(X))$. Hence, for some $0<a<b$, $\Lambda\in\B_\mathrm{c}(X)$, $N\in\N$, \eqref{Eq4}  holds for all $\eta\in A$. Therefore, for each
  $\eta\in\K(X)$, 
  \begin{align}
    |(KG)(\eta)|&\leq c(K\i_A)(\eta)=c\sum_{k=0}^N\sum_{\substack{\xi\subset  p_{\Lambda,a,b}(\eta) \\|\tau(\xi)|=k}}\i_A(\xi)\notag \\&\leq c\sum_{k=0}^N\binom{|\tau(  p_{\Lambda,a,b}(\eta) )|}{k}\leq C(1+\eta(\Lambda))^N,\label{polbddpd}
  \end{align}
  for $C=c (\max\bigl\{1,\frac{1}{a}\bigr\})^N$, where we used the estimate \eqref{eq00}.
  \end{proof}
  
We can also extend the $K$-transform to the class $\cF_{\exp}(\K_0(X))$ of measurable functions $G:\K_0(X)\to\R$ such that, for some $\Lambda\in\B_\mathrm{c}(X)$ and $C>0$,
\begin{equation}
  |G(\xi)|\leq \i_{\{\tau(\xi)\subset\Lambda\}}(\xi)\, C^{|\tau(\xi)|}\prod_{x\in\tau(\xi)}s_x,\qquad \xi\in\K_0(X).\label{Eqr8}
\end{equation}
Indeed, for each $\eta\in\K(X)$, we have then
\[
|(KG)(\eta)|\leq\sum_{\xi\Subset\eta}|G(\xi)|\leq\!
\sum_{\substack{\xi\Subset\eta:\\\tau(\xi)\subset\tau(\eta)\cap\Lambda}}\!C^{|\tau(\xi)|}\prod_{x\in\tau(\xi)}s_x
=\!\prod_{x\in\tau(\eta)\cap\Lambda}(1+Cs_x)<\infty,
\]
since $\eta\in\K(X)$ and hence, $\sum\limits_{x\in\tau(\eta)\cap\Lambda}Cs_x = C\eta(\Lambda)<\infty$.

\begin{example}
Let $f:X\to\R$ be a bounded measurable function with compact support.

\noindent
1.~For $G\in\A$ defined by
\[
G(\eta):=\begin{cases}
sf(x),&\text{ if } \eta=\{s\delta_x\}\in\K_0^{(1)}(X)\\
0,&\text{ otherwise}\end{cases},\quad\eta\in\K_0(X),
\]
the $K$-transform of $G$ is given by
\[
(KG)(\eta)=\sum_{x\in\tau(\eta)}s_xf(x)=\langle\eta,f\rangle,\quad\eta\in\K(X).
\]
2.~For the so-called Lebesgue--Poisson exponent $e_\K(f)\in\A$
corresponding to $f$,
\begin{equation}\label{LPexpdef}
  e_{\K}(f,\eta):=\prod_{x\in\tau(\eta)}s_xf(x),\quad \eta\in\K_0(X),
\end{equation}
its $K$-transform is equal to
\begin{equation}\label{Kcohst}
  (Ke_\K(f))(\eta)=\prod_{x\in\tau(\eta)}(1+s_xf(x)),\quad\eta\in\K(X).
\end{equation}
%\end{enumerate}
\end{example}

Given $G_1$, $G_2\in L^0(\K_0(X))$, let us define the $\star$-convolution between $G_1$ and $G_2$,
\[
(G_1\star G_2)(\eta):=\sum_{\substack{\xi_1+\xi_2+\xi_3=\eta\\\tau(\xi_i)\cap\tau(\xi_j)=\emptyset,\,i\not=j}}G_1(\xi_1+\xi_2)G_2(\xi_2+\xi_3),\quad\eta\in\K_0(X),
\]
where the sum is over all $\xi_1,\xi_2,\xi_3\subset\eta$ such that $(\tau(\xi_1), \tau(\xi_2), \tau(\xi_3))$ is a
partition of $\tau(\eta)$. As easily seen, under this product $L^0(\K_0(X))$ has a commutative
algebraic structure with unit element $e_\K(0)$. 

\begin{proposition}
For all $G_1,G_2\in L^0_{\mathrm{ls}}(\K_0(X))$ we have $G_1\star G_2\!\in\! L^0_{\mathrm{ls}}(\K_0(X))$ and
\begin{equation}
K(G_1\star G_2)=(KG_1)\cdot(KG_2).\label{fpc}
\end{equation}
\end{proposition}

\begin{proof}
Given $G_1, G_2\in L^0_{\mathrm{ls}}(\K_0(X))$ we have $G_i=G_i\i_{A}$ for some $A\in\B(\K_0(X))$ such that \eqref{Eq41} holds. Then
\[
(G_1\star G_2)\i_A=((G_1\i_{A})\star(G_2\i_{A}))\i_A=((G_1\i_{A})\star(G_2\i_{A}))=G_1\star G_2.
\]
This shows that $G_1\star G_2\in L^0_{\mathrm{ls}}(\K_0(X))$. Concerning the right-hand side of \eqref{fpc}, 
\begin{equation}
(KG_1)(\eta)\cdot(KG_2)(\eta)=\left(\sum_{\xi\Subset\eta}G_1(\xi)\right)\left(\sum_{\zeta\Subset\eta}G_2(\zeta)\right),\ \eta\in\K(X),\label{Eq6}
\end{equation}
observe that, for each $\eta\in\K(X)$ fixed, there is a one-to-one correspondence between pairs $\xi\Subset\eta$,
$\zeta\Subset\eta$ and groups $\vartheta\Subset\eta$, $\xi_1,\xi_2,\xi_3\subset\vartheta$ with $\xi_1+\xi_2+\xi_3=\vartheta$, hence, 
$(\tau(\xi_1), \tau(\xi_2), \tau(\xi_3))$ forms a partition of $\tau(\vartheta)$. This one-to-one correspondence is
defined by the following rule: $\tau(\vartheta)=\tau(\xi)\cup\tau(\zeta)$,
$\tau(\xi_1)=\tau(\xi)\setminus\tau(\zeta)$, $\tau(\xi_2)=\tau(\xi)\cap\tau(\zeta)$,
$\tau(\xi_3)=\tau(\zeta)\setminus\tau(\xi)$. In this way, product \eqref{Eq6} can be rewritten as
\[
\sum_{\vartheta\Subset\eta}\sum_{\xi_1+\xi_2+\xi_3=\vartheta}G_1(\xi_1+\xi_2)G_2(\xi_2+\xi_3),
\]
which completes the proof.
\end{proof}

\subsection{Correlation measures and correlation functions on $\K_0(X)$}

A measure $\rho$ on $\bigl(\K_0(X),\B(\K_0(X))\bigr)$ is said to be \emph{locally finite} if $\rho(A)<\infty$ for each $A\in\B_\mathrm{b}(\K_0(X))$. 

\begin{example}
  An example of a locally finite measure is the Lebesgue--Poisson measure $\lambda_{\K_0,\nu\otimes\sigma}$, where $\nu$ and $\sigma$ are non-atomic positive Radon measures on $\mathcal{B}(\R^*_+)$
and $\mathcal{B}(X)$, respectively, and $\nu$ has a finite first moment \eqref{fmoment}. The measure $\lambda_{\K_0,\nu\otimes\sigma}$ is defined so that, for each $G\in B_{\mathrm{bs}}(\K_0(X))$,
\begin{align}
  &\int_{\K_0(X)}G(\eta)\,d\lambda_{\K_0,\nu\otimes\sigma}(\eta)\notag\\
  &=G(0)+\sum_{n=1}^\infty\frac{1}{n!}\int_{\X^n}G\left(\sum_{k=1}^ns_i\delta_{x_i}\right)d\nu(s_1)\dotsm d\nu(s_n)d\sigma(x_1)\dotsm d\sigma(x_n).\label{Eq8}
\end{align} 
Note that the sum in \eqref{Eq8} is finite for $G\in B_{\mathrm{bs}}(\K_0(X))$, in particular, for $G=\i_A$ with $A\in\B_\mathrm{b}(\K_0(X))$. Moreover, it is easily seen that $\lambda_{\K_0,\nu\otimes\sigma}$ is the push-forward of $\lambda_{\nu\otimes\sigma}\bigr\vert_{\mathcal{B}(\Pi_0(\X))}$ under $\mathcal{R}\bigr\vert_{|\Pi_0(\X)}$ to $\bigl(\K_0(X),\B(\K_0(X))\bigr)$, where $\lambda_{\nu\otimes\sigma}$ is the Lebesgue--Poisson on $\bigl(\Gamma_0(\X),\B(\Gamma_0(\X))\bigr)$, see \cite{harmana}. In particular, we can extend \eqref{Eq8} to measurable functions on $\K_0(X)$ integrable with respect to  $\lambda_{\K_0,\nu\otimes\sigma}(\eta)$ for which the right-hand side  of \eqref{Eq8} is well-defined. For example, if $f\in L^1(X,\sigma)$
and $e_\K(f)$ is defined as in \eqref{LPexpdef} for $\lambda_{\K_0,\nu\otimes\sigma}$-a.a.~$\eta\in\K_0(X)$, 
then
$e_\K(f)\in L^1(\K_0(X),\lambda_{\K_0,\nu\otimes\sigma})$ with
\[
\int_{\K_0(X)}e_\K(f,\eta)\,d\lambda_{\K_0,\nu\otimes\sigma}(\eta)
=\exp\left(\int_{\R^*_+}s\,d\nu(s)\int_Xf(x)\,d\sigma(x)\right).
\]
\end{example}

\begin{definition}\label{def:corm}
Let $\mu\in\mathcal{M}^1_{\mathrm{fm}}(\K(X))$ be given. We (uniquely) define a measure $\rho_\mu$ on on $(\K_0(X),\B(\K_0(X))$ by requiring that, for all $A\in\B_\mathrm{b}(\K_0(X))$,
\[
\rho_\mu(A)=\int_{\K(X)}(K\i_A)(\eta)\,d\mu(\eta).
\]
Then $\rho_\mu$  is called the correlation measure corresponding to $\mu$.  
\end{definition}

\begin{remark}
Note that, by \eqref{polbddpd}, the assumption $\mu\in\mathcal{M}^1_{\mathrm{fm}}(\K(X))$
ensures that $\rho_\mu(A)<\infty$ for all $A\in\B_\mathrm{b}(\K_0(X))$. 
\end{remark}

\begin{remark}
  It can be easily derived from \cite[Theorem 2 and formula (3.9)]{Len1973} that, under a very week assumption, the correlation measure $\rho_\mu$ in Definition~\ref{def:corm} uniquely determines the measure $\mu$.
\end{remark}

\begin{proposition}\label{lacormeasure}
  Let $\nu$ and $\sigma$ be two non-atomic positive Radon measures on $\mathcal{B}(\R^*_+)$
  and $\mathcal{B}(X)$, respectively. Assume that $\nu$ has finite moments of all orders \eqref{allmoments}, so that
  $\pi_{\K,\nu\otimes\sigma}\in\mathcal{M}^1_{\mathrm{fm}}(\K(X))$ by Proposition~\ref{Proposition5}. Then, $\lambda_{\K_0,\nu\otimes\sigma}$ is the correlation measure of $\pi_{\K,\nu\otimes\sigma}$.
\end{proposition}
\begin{proof}
  Let $A\in\B_\mathrm{b}(\K_0(X))$ and \eqref{Eq4} holds. 
  Then, by \eqref{5Eq1}, $(K\i_A)(\eta)
  =(K\i_A)(p_{\Lambda,a,b}(\eta))$ and hence,
  \[
    \int_{\K(X)} (K\i_A)(\eta)d\pi_{\K,\nu\otimes\sigma}(\eta) = \int_{\K([a,b]\times\Lambda)} (K\i_A)(\eta)d\pi^{\Lambda,a,b}_{\K,\nu\otimes\sigma}(\eta).
  \]
Then, by \eqref{eqd2qwr2413} and Proposition~\ref{support} with $\hat{X}$ replaced by $[a,b]\times\Lambda$, we get that
\begin{align*}
  &\quad \int_{\K([a,b]\times\Lambda)} (K\i_A)(\eta)d\pi^{\Lambda,a,b}_{\K,\nu\otimes\sigma}(\eta)\\&=
  \int_{\Gamma([a,b]\times\Lambda)} (K_{\hat{X}}\i_{\mathcal{R}^{-1}A})(\gamma)d\pi^{[a,b]\times \Lambda}_{\nu\otimes\sigma}(\gamma)\\
  &=
  e^{-\nu([a,b])\sigma(\Lambda)}\int_{\Gamma([a,b]\times\Lambda)} \sum_{\xi\subset\gamma}\i_{\mathcal{R}^{-1}A}(\xi) d\lambda_{\nu\otimes\sigma}(\gamma)
  \intertext{and rewriting $\i_{\mathcal{R}^{-1}A}(\xi)=\i_{\mathcal{R}^{-1}A}(\xi)\cdot\i_{\Gamma([a,b]\times\Lambda)}(\gamma\setminus\xi)$, we can apply e.g. \cite[Lemma~A.1]{harmana}}&=
  e^{-\nu([a,b])\sigma(\Lambda)}\int_{\Gamma([a,b]\times\Lambda)} \i_{\mathcal{R}^{-1}A}(\xi) d\lambda_{\nu\otimes\sigma}(\xi)\int_{\Gamma([a,b]\times\Lambda)} d\lambda_{\nu\otimes\sigma}(\gamma) \\&=\lambda_{\nu\otimes\sigma}({\mathcal{R}^{-1}A})=\lambda_{\K_0,\nu\otimes\sigma}(A),
\end{align*}
which proves the statement.
\end{proof}

Let $\mu\in\mathcal{M}^1_{\mathrm{fm}}(\K(X))$. Then, by \eqref{polbddpd}, $K(B_{\mathrm{bs}}(\K_0(X)))\subset L^1(\K(X),\mu)$. Thus,
$B_{\mathrm{bs}}(\K_0(X))\subset L^1(\K_0(X),\rho_\mu)$ and standard techniques of the measure theory yield
\begin{equation}
\int_{\K_0(X)}G(\eta)\,d\rho_\mu(\eta)=\int_{\K(X)}(KG)(\eta)\,d\mu(\eta),\quad G\in B_{\mathrm{bs}}(\K_0(X)).\label{Eq9}
\end{equation}

The density of $B_{\mathrm{bs}}(\K_0(X))$ in $L^1(\K_0(X),\rho_\mu)$ allows to
extend the $K$-trans\-form to a bounded operator. We will keep the same notation for the extended operator. More
precisely, we have the following result.

\begin{proposition}\label{Proposition4}
Let $\mu\in\mathcal{M}^1_{\mathrm{fm}}(\K(X))$. Then, there is a bounded operator $K:L^1(\K_0(X),\rho_\mu)\to L^1(\K(X),\mu)$
such that formula \eqref{Eq9} holds for any $G\in L^1(\K_0(X),\rho_\mu)$. Moreover, for each
$G\in L^1(\K_0(X),\rho_\mu)$, equality \eqref{Eq3} holds for $\mu$-almost all $\eta\in\K(X)$.  
\end{proposition}

\begin{proof}
The proof is based on standard techniques of measure theory and follows similarly to
\cite[Corollary 4.1 and Theorem 4.1]{harmana}.
\end{proof}

\begin{remark}
  In particular, for any $\mu\in\mathcal{M}^1_{\mathrm{fm}}(\K(X))$ and $f:X\to\R$ such that $e_\K(f)\in L^1(\K_0(X),\rho_\mu)$, it follows from Proposition~\ref{Proposition4} that \eqref{Kcohst} holds for $\mu$-a.a.~$\eta\in\K(X)$.
\end{remark}

\begin{proposition}\label{abscont}
  Let the conditions of Proposition~\ref{lacormeasure} hold. Let $\mu\in\mathcal{M}^1_{\mathrm{fm}}(\K(X))$ be locally absolutely continuous with respect to the
  measure $\pi_{\K,\nu\otimes\sigma}$. Let $\rho_\mu$ be the correlation measure of $\mu$ according to Definition~\ref{def:corm}. Then
  $\rho_\mu$ is absolutely continuous with respect to $\lambda_{\K_0,\nu\otimes\sigma}$. 
  \end{proposition}
    \begin{proof}
  Given $A\in\mathcal{B}_\mathrm{b}(\K_0(X))$, assume that $\lambda_{\K_0,\nu\otimes\sigma}(A)=0$. Hence, for some $0<a<b$,
  $N\in\N$, $\Lambda\in\B_\mathrm{c}(X)$, \eqref{Eq4} holds for all $\eta\in A$ and we have
  \begin{align*}
  0&=\lambda_{\K_0,\nu\otimes\sigma}(A)
  \\&=\int_{\K(X)}\left(K\i_A\right)(\eta)d\pi_{\K,\nu\otimes\sigma}(\eta)\\
  &=\int_{\K(X)}\left(K\i_A\right)(p_{\Lambda,a,b}(\eta))d\pi_{\K,\nu\otimes\sigma}(\eta)=\int_{\K([a,b]\times \Lambda)}\left(K\i_A\right)(\eta)d\pi^{\Lambda,a,b}_{\K,\nu\otimes\sigma}(\eta).
  \end{align*}
  This implies that $K\i_A=0$ $\pi^{\Lambda,a,b}_{\K,\nu\otimes\sigma}$-a.e.~on $\K([a,b]\times\Lambda)$. As a result,
 \begin{align*}
  \rho_\mu(A)&=\int_{\K(X)}\left(K\i_A\right)(\eta)d\mu(\eta)\\&=\int_{\K([a,b]\times\Lambda)}\left(K\i_A\right)(\eta)\frac{d\mu^{\Lambda,a,b}}{d\pi^{\Lambda,a,b}_{\K,\nu\otimes\sigma}}(\eta)d\pi^{\Lambda,a,b}_{\K,\nu\otimes\sigma}(\eta)=0.\qedhere
 \end{align*}
  \end{proof}

\begin{definition}
Let the conditions of Proposition \ref{abscont} hold. The Radon--Nikodym derivative 
$k_\mu:=\dfrac{d\rho_\mu}{d\lambda_{\K_0,\nu\otimes\sigma}}$
is called the correlation function corresponding to $\mu$.
\end{definition}

\section{Finite-difference calculus on the cone}

\subsection{Discrete gradients on $\K(X)$}

The discrete structure of measures in $\K(X)$ suggests a development of a finite-difference calculus on $\K(X)$. For
$\eta\in\K(X)$, the elementary operations on $\eta$ which will be considered are the following ones:
\begin{itemize}
\item removing one point $x$ from the support of $\eta$: $\eta\mapsto\eta-s_x\delta_x$;
\item adding a new point $x$ with a weight $s$ to $\eta$: $\eta\mapsto\eta+s\delta_x$, $s\in\R^*_+$.
\end{itemize}

As a set of test functions on $\K(X)$ we will consider the space $\F:=\mathcal{F}C_\mathrm{b}(C_\mathrm{c}(\X),\mathbb{K}(X))$ of all functions $F:\K(X)\to\R$
of the form
\[
F(\eta)=g\left(\langle\mathcal{R}^{-1}\eta,\varphi_1\rangle,\dotsc,\langle\mathcal{R}^{-1}\eta,\varphi_N\rangle\right),\quad\eta\in\K(X),
\]
where $g\in C_\mathrm{b} (\R^N)$, $\varphi_1,\dotsc,\varphi_N\in C_\mathrm{c} (\X)$, $N\in\N$. We fix
two non-atomic positive Radon measures $\nu$ and $\sigma$ on $\mathcal{B}(\R^*_+)$ and $\mathcal{B}(X)$,
respectively, such that $\nu$ has finite first moment.

\begin{definition}
Let $F\in\F$.
%\begin{enumerate}

1.~A discrete death gradient of $F$ is defined by
\[
(D_x^-F)(\eta):=F(\eta-s_x\delta_x)-F(\eta),\quad\eta\in\K(X),x\in\tau(\eta).
\]  
The corresponding tangent space is chosen to be $T^-_\eta(\K(X)):=L^2(X,\eta)$.

2.~A discrete birth gradient of $F$ is defined by
\[
(D_{(s,x)}^+F)(\eta):=F(\eta+s\delta_x)-F(\eta),\quad\eta\in\K(X),
\]  
where $(s,x)\in\X$, $x\notin\tau(\eta)$. Here, the corresponding tangent space is chosen to be
$T^+_\eta(\K(X)):=L^2(\X,s\nu(ds)\otimes\sigma(dx))$.
%\end{enumerate}
\end{definition}

Observe that, for a fixed $\eta\in\K(X)$, the function $\tau(\eta)\ni x\mapsto (D_x^-F)(\eta)$ is bounded and has 
compact support in $X$. Thus, $(D_\cdot^-F)(\eta)\in T^-_\eta(\K(X))$. For each $h\in C_\mathrm{c}(X)$, we define a
directional derivative along $h$ by
\begin{align*}
(D_h^-F)(\eta):&=\langle(D_\cdot^-F)(\eta),h\rangle_{T_\eta^-(\K(X))}\\
&=\int_{X}(D_x^-F)(\eta)h(x)d\eta(x)=\sum_{x\in\tau(\eta)}s_xh(x)(D_x^-F)(\eta).
\end{align*}
As easily seen, for each $\eta\in\K(X)$ fixed, we also have $(D_\cdot^+F)(\eta)\in T^+_\eta(\K(X))$ and we define a
directional derivative along a direction $h\in C_\mathrm{c}(X)$ by
\begin{align*}
  (D_h^+F)(\eta):&=\langle(D_\cdot^+F)(\eta),h\rangle_{T_\eta^+(\K(X))}\\
&=\int_X\int_{\R^*_+}(D_{(s,x)}^+F)(\eta)sd\nu(s)h(x)d\sigma(x).
\end{align*}
The next result states a relation between these two notions of directional derivative.

\begin{proposition}
For any $F,G\in\F$ and $h\in C_\mathrm{c}(X)$ we have
\begin{align*}
&\int_{\K(X)}(D_h^-F)(\eta)G(\eta)\,d\pi_{\K,\nu\otimes\sigma}(\eta)\\
&=\int_{\K(X)}F(\eta)(D_h^+G)(\eta)\,d\pi_{\K,\nu\otimes\sigma}(\eta)
-\int_{\K(X)}F(\eta)G(\eta)B_{\nu,\sigma,h}(\eta)\,d\pi_{\K,\nu\otimes\sigma}(\eta),
\end{align*}
where
\[
B_{\nu,\sigma,h}(\eta):=\int_Xh(x)\,d\eta(x)-\int_{\R^*_+}s\,d\nu(s)\int_Xh(x)\,d\sigma(x).
\]
\end{proposition}

\begin{proof}
By formula \eqref{yfuftyttfyoi9u}, 
\begin{align*}
&\int_{\K(X)}\int_X F(\eta-s_x\delta_x)h(x)d\eta(x)G(\eta)\,d\pi_{\K,\nu\otimes\sigma}(\eta)\\
&=\int_{\K(X)}\int_X\int_{\R^*_+}sF(\eta)h(x)G(\eta+s\delta_x)\,d\nu(s)d\sigma(x)d\pi_{\K,\nu\otimes\sigma}(\eta)\\
&=\int_{\K(X)}F(\eta)(D_h^+G)(\eta)\,d\pi_{\K,\nu\otimes\sigma}(\eta)
\\&\quad +\int_{\K(X)}F(\eta)G(\eta)\,d\pi_{\K,\nu\otimes\sigma}(\eta)\int_{\R^*_+}s\,d\nu(s)\int_Xh(x)\,d\sigma(x).
\end{align*}
Therefore, by the definition of $D_h^-F$, the required formula follows.
\end{proof}

We proceed to show that a Laplacian-type operator associated with the discrete birth-and-death gradients exists.
For each $F,G\in\F$, let $\mathcal{E}$ be the Dirichlet integral associated with the discrete death gradient:
\begin{align*}
\mathcal{E}(F,G)&:=\int_{\K(X)}\langle (D_\cdot^-F)(\eta),(D_\cdot^-G)(\eta)\rangle_{T^-_\eta(\K(X))}\,d\pi_{\K,\nu\otimes\sigma}(\eta)\\
&=\int_{\K(X)}\int_X (D_x^-F)(\eta)(D_x^-G)(\eta)\,d\eta(x)d\pi_{\K,\nu\otimes\sigma}(\eta).
\end{align*}
It turns out by formula \eqref{yfuftyttfyoi9u} that, actually, $\mathcal{E}$ coincides with the Dirichlet integral
associated with the discrete birth gradient:
\begin{align}
\mathcal{E}(F,G)&:=\int_{\K(X)}\langle (D_\cdot^+F)(\eta),(D_\cdot^+G)(\eta)\rangle_{T^+_\eta(\K(X))}\,d\pi_{\K,\nu\otimes\sigma}(\eta)\nonumber \\
&=\int_{\K(X)}\int_X\int_{\R^*_+} s(D_{(s,x)}^+F)(\eta)(D_{(s,x)}^+G)(\eta)\,d\nu(s)d\sigma(x)d\pi_{\K,\nu\otimes\sigma}(\eta).\label{1Dec}
\end{align}

\begin{proposition}\label{propZh1}
The $(\mathcal{E},\F)$ is a well-defined symmetric bilinear form on $L^2(\K(X),\pi_{\K,\nu\otimes\sigma})$.  
\end{proposition}

\begin{proof}
The symmetry and the bilinear property follow directly. Hence, we just have to show that if $F\in\F$ is
$\pi_{\K,\nu\otimes\sigma}$-a.e.~equal to 0, then $\mathcal{E}(F,G)=0$ for every $G\in\F$. This is a
consequence of formula \eqref{yfuftyttfyoi9u}, because for each $\Lambda\in\mathcal{B}_\mathrm{c}(X)$ one
then finds
\begin{align*}
&\int_{\K(X)}\int_X\int_{\R^*_+}s|F(\eta+s\delta_x)|\i_\Lambda(x)\,d\nu(s)d\sigma(x)d\pi_{\K,\nu\otimes\sigma}(\eta)\\
&=\int_{\K(X)}\int_X|F(\eta)|\eta(\Lambda)\,d\pi_{\K,\nu\otimes\sigma}(\eta)=0,
\end{align*}
which implies that $F(\eta+s\delta_x)=0$ for
$\pi_{\K,\nu\otimes\sigma}\otimes\nu\otimes\sigma$-a.a.~$(\eta,s,x)\in\K(X)\times\X$. Thus, $(D^+_{(s,x)} F)(\eta)=0$
for $\pi_{\K,\nu\otimes\sigma}\otimes\nu\otimes\sigma$-a.a.~$(\eta,s,x)\in \K(X)\times\R^*_+\times X$. Hence, by \eqref{1Dec}, for each $G\in\F$ we have $\mathcal{E}(F,G)=0$.
\end{proof}

\begin{proposition}\label{propZh2}
For each $F\in\F$, let
\[
(LF)(\eta):=\int_X(D^-_xF)(\eta)\,d\eta(x)+\int_X\int_{\R^*_+}(D_{(s,x)}^+F)(\eta)sd\nu(s)d\sigma(x).
\]
Then, $(L,\F)$ is a symmetric operator on $L^2(\K(X),\pi_{\K,\nu\otimes\sigma})$ which verifies the following
equality
\begin{equation}
\mathcal{E}(F,G)=\langle-LF,G\rangle_{L^2(\K(X),\pi_{\K,\nu\otimes\sigma})},\quad F,G\in\F.\label{26Eq1}
\end{equation}
The bilinear form $(\mathcal{E},\F)$ is closable on $L^2(\K(X),\pi_{\K,\nu\otimes\sigma})$ and the
operator $(L,\F)$ has Friedrich's extension, denoted by $(L,D(L))$. Moreover, the extended operator $(L,D(L))$ is the
generator of the closed symmetric form, denoted by $(\mathcal{E},D(\mathcal{E}))$. 
\end{proposition}

\begin{proof}
First note that, by formula \eqref{yfuftyttfyoi9u}, for any $F,G\in\F$ we have
\begin{align*}
\mathcal{E}(F,G)&=\int_{\K(X)}(D^-_xF)(\eta)G(\eta-s_x\delta_x)\,d\eta(x)d\pi_{\K,\nu\otimes\sigma}(\eta)\\
&\quad-\int_{\K(X)}(D^-_xF)(\eta)G(\eta)\,d\eta(x)d\pi_{\K,\nu\otimes\sigma}(\eta)\\
&=-\int_{\K(X)}\int_X\int_{\R^*_+}sF(\eta)(D^+_{(s,x)}G)(\eta)\,d\nu(s)\sigma(x)d\pi_{\K,\nu\otimes\sigma}(\eta)\\
&\quad-\int_{\K(X)}\int_XF(\eta)(D^-_xG)(\eta)\,d\eta(x)d\pi_{\K,\nu\otimes\sigma}(\eta),
\end{align*}
which, due to the symmetry of $\mathcal{E}$, shows that formula \eqref{26Eq1} holds, provided
$LF\in L^2(\K(X),\pi_{\K,\nu\otimes\sigma})$. In order to prove that $LF\in L^2(\K(X),\pi_{\K,\nu\otimes\sigma})$,
observe that since $F\in\F$, there are $C\geq 0$, $0<a<b$, $\Lambda\in\mathcal{B}_\mathrm{c}(X)$ such that
\begin{align*}
&|(D^-_xF)(\eta)|\leq C\i_{\left[a,b\right]}(s_x)\i_\Lambda(x),\quad x\in\tau(\eta),\eta\in\K(X),\\[2mm]
&|(D^+_{(s,x)}F)(\eta)|\leq C\i_{\left[a,b\right]}(s)\i_\Lambda(x),\quad s\in\R^*_+,x\in X\setminus\tau(\eta),\eta\in\K(X).
\end{align*}
Thus, 
\begin{align}\label{26Eq2}
&\quad \int_{\K(X)}(LF)^2(\eta)\,d\pi_{\K,\nu\otimes\sigma}(\eta)\\\notag
&\leq 2C^2\int_{\K(X)}\left(\int_X\i_{\left[a,b\right]}(s_x)\i_\Lambda(x)\,d\eta(x)\right)^2d\pi_{\K,\nu\otimes\sigma}(\eta) \\
&\quad +2C^2\int_{\K(X)}\left(\int_X\i_\Lambda(x)\,d\sigma(x)\int_{\R^*_+}s\i_{\left[a,b\right]}(s)\,d\nu(s)\right)^2d\pi_{\K,\nu\otimes\sigma}(\eta),\notag
\end{align}
where the latter integral is finite. Concerning \eqref{26Eq2}, three applications of formula \eqref{yfuftyttfyoi9u}
yield
\begin{align*}
&\int_{\K(X)}\left(\int_X\i_{\left[a,b\right]}(s_x)\i_\Lambda(x)\,d\eta(x)\right)^2d\pi_{\K,\nu\otimes\sigma}(\eta)\\
&=\int_{\K(X)}\int_X\int_X\!\i_{\left[a,b\right]}(s_x)\i_\Lambda(x)\i_{\left[a,b\right]}(s_y)\i_\Lambda(y)d\eta(x)d(\eta-s_x\delta_x)(y) d\pi_{\K,\nu\otimes\sigma}(\eta)\\
&+\int_{\K(X)}\int_Xs_x\i_{\left[a,b\right]}(s_x)\i_\Lambda(x)\,d\eta(x)d\pi_{\K,\nu\otimes\sigma}(\eta)\\
&=\int_{\K(X)}\left(\int_X\i_\Lambda(x)\,d\sigma(x)\int_{\R^*_+}s\i_{\left[a,b\right]}(s)\,d\nu(s)\right)^2d\pi_{\K,\nu\otimes\sigma}(\eta)\\
&+\int_{\K(X)}\int_X\int_{\R^*_+}s^2\i_{\left[a,b\right]}(s)\i_\Lambda(x)\,d\nu(s)d\sigma(x)d\pi_{\K,\nu\otimes\sigma}(\eta)<\infty. 
\end{align*}
Hence, $LF\in L^2(\K(X),\pi_{\K,\nu\otimes\sigma})$. Furthermore, by formula \eqref{26Eq1}, $(-L,\F)$ is a positive
symmetric operator in $L^2(\K(X),\pi_{\K,\nu\otimes\sigma})$. It is now standard to prove that the bilinear form
$(\mathcal{E},\F)$ is closable and $(L,\F)$ has Friedrich's extension (see e.g.~\cite{RS}).
\end{proof}

\begin{remark}
  Using techniques of the Dirichlet form theory \cite{MR1992}, it is possible to show that there exists an equilibrium Markov process on $\K(X)$ that has the operator $L$ as its generator, compare with \cite{KL2005} and \cite{CKL2016}.
\end{remark}

\subsection{Polynomial functions on $\K(X)$}

Let $n\in\N$. Let $\M^{(n)}(X)$ be the set of all symmetric real-valued Radon measures on $\bigl(X^n,\B(X^n)\bigr)$, see \cite{FKLO21}. Let $\cF(X)$ be the set of all bounded measurable functions on $X$ with compact support, and $\cF^{(n)}(X)$ the set of all bounded measurable symmetric functions on $X^n$ with compact support. For $\mu^\n\in \M^\n(X)$ and $f^\n\in\cF^\n(X)$, we denote 
\[
  \langle \mu^\n,f^\n\rangle:=\int_{X^n}f^\n d\mu^\n, \quad n\in\N.
\]

Let $\eta=\sum\limits_i s_{x_i}\delta_{x_i}\in\K(X)$. We
set  $P^{(0)}(\eta):=1$ and $P^{(1)}(\eta):=\eta\in\M(X)=\M^{(1)}(X)$, and consider the measure $P^{(n)}(\eta)$ on $(X^n,\B(X^n))$ for $n\geq 2$, given by
\begin{align}\label{fyd6e4}
  &\quad P^{(n)}(\eta)(dx_1\dotsm dx_n)\\&:=
\eta(dx_1)(\eta(dx_2)-s_{x_1}\delta_{x_1}(dx_2))\notag \times\ldots\times\\
&\quad \times 
(\eta(dx_n)-s_{x_1}\delta_{x_1}(dx_n)-s_{x_2}\delta_{x_2}(dx_n)-\ldots -s_{x_{n-1}}\delta_{x_{n-1}}(dx_n)).\notag
\end{align}
It is straightforward to check that \eqref{fyd6e4} does not depend on the ordering in $\eta=\sum\limits_i s_{x_i}\delta_{x_i}$, therefore, $P^{(n)}(\eta)$ is a symmetric measure on $X^n$. Moreover, 
\begin{align} \label{fyd6e2}
  P^{(n)}(\eta)&=\sum_{x_1\in\tau(\eta)}\sum_{x_2\in\tau(\eta)\setminus\{x_1\}}\ldots\sum_{x_n\in\tau(\eta)\setminus\{x_1,\ldots,x_{n-1}\}}s_{x_1}s_{x_2}\ldots s_{x_n}\\& \qquad\qquad \qquad\qquad\qquad \times \delta_{x_1}\otimes\delta_{x_2}\otimes\ldots\otimes\delta_{x_n}\notag\\
  &= n! \sum_{\{x_1,\ldots,x_n\}\subset \tau(\eta)}s_{x_1}s_{x_2}\ldots s_{x_n}\delta_{x_1}\odot\delta_{x_2}\odot\ldots\odot\delta_{x_n},
 \notag
\end{align}
where $\odot$ denotes the tensor product symmetrization, cf. \cite{FKLO21}. Then, for any $f\in\cF(X)$,
\begin{equation}\label{eq:q32}
  \bigl\lvert\langle P^\n(\eta),f^{\otimes n}\rangle\bigr\rvert \leq  \langle \eta, |f|\rangle^n<\infty, \quad n\in\N.
\end{equation}
By the polarization identity, any $\mu^\n\in\M^\n(X)$ is uniquely defined by the values of $\langle \mu^\n,f^{\otimes n}\rangle$ for $f\in\cF(X)$. Therefore, $P^\n(\eta)\in\M^\n(X)$, $n\in\N$.

Let now $n\in\N$ and $f^{(n)}\in \cF^{(n)}(X)$. We define the following \emph{polynomial function} on $\K(X)$
\begin{equation*}
  p_n(\eta):=  \langle P^\n(\eta),f^\n\rangle.
\end{equation*}

\begin{remark}
  We stress that $p_n$ is not the restriction of a polynomial on $\M(X)$ (in the sense of \cite{FKLO21}) to $\K(X)$ as the right-hand side  of \eqref{fyd6e2} may not be even defined for an arbitrary $\eta\in\M(X)\setminus\K(X)$.
\end{remark}

We define also the polynomial sequence  $\{(\omega)_n: n\geq0 \}$ of \emph{falling factorials} on $\M(\X)$, cf. \cite{FKLO16,FKLO21}. Namely, we set $(\omega)_0:=1$, $(\omega)_1:=\omega$; and, for $n\geq2$ and $\hat{y}_i:=(s_i,x_i)$, $1\leq i\leq n$, we define
\begin{align}\label{fyd6e}
  &\quad (\omega)_n(d\hat{y}_1\ldots d\hat{y}_n)\\:&=
\omega(d\hat{y}_1)(\omega(d\hat{y}_2)-\delta_{\hat{y}_1}(d\hat{y}_2))\notag\\
&\qquad\ \times\ldots\times
(\omega(d\hat{y}_k)-\delta_{\hat{y}_1}(d\hat{y}_k)-\delta_{\hat{y}_2}(d\hat{y}_k)-\ldots -\delta_{\hat{y}_{n-1}}(d\hat{y}_n))\notag\\&=
n! \sum_{\{\hat{y}_1,\ldots,\hat{y}_n\}\subset\gamma}\delta_{\hat{y}_1}\odot\delta_{\hat{y}_2}\odot\ldots\odot\delta_{\hat{y}_n}.\notag
\end{align}
By \cite{FKLO16,FKLO21}, the generating function of falling factorials on $\M(\X)$ is 
\begin{equation}
  \sum_{n=0}^\infty\frac{1}{n!}\langle(\omega)_n,\hat{f}^{\otimes n}\rangle=\exp\bigl(\langle\omega,\log(1+\hat{f})\rangle\bigr),\qquad \hat{f}\in\cF(\X),\label{31Eq2}
  \end{equation}
that is understood as an equality of formal power series, see \cite[Subsection 2.2 and Appendix]{FKLO16}. Moreover, by \cite{FKLO16}, the falling factorials are of binomial type, i.e.,
\begin{equation}
  (\omega+\omega')_n=\sum_{k=0}^n\binom{n}{k}(\omega)_k\odot(\omega')_{n-k},\quad\omega,\omega'\in\M(\X),n\in\N,\label{bin-id-marked}
\end{equation}
and the following lowering property holds: for each $n\in\N_0$, $\hat{y}\in\X$, $\omega\in\M(\X)$,
\begin{equation}\label{lowering}
  (\omega+\delta_{\hat{y}})_n-(\omega)_n=n\delta_{\hat{y}}\odot(\omega)_{n-1}.
\end{equation}

Let $f\in\cF(X)$ and consider, for each $j\in\N$,
\[
\hat{f}_j(s,x):=\i_{[\frac{1}{j},j]}(s)\, s\, f(x), \quad (s,x)\in\X.  
\]
Then $\hat{f}_j\in\cF(\X)$ and we can consider
$\langle(\omega)_n,\hat{f}_j^{\otimes n}\rangle$, $\omega\in\M(\X)$. Let $\hat{f}(s,x):=sf(x)$ for $(s,x)\in\X$. Then, by \eqref{fyd6e}, for each $\gamma\in\Pi(\X)\subset\M(\X)$, we have, cf. \eqref{eq:q32},
\[
\bigl\lvert \langle (\gamma)_n, \hat{f}^{\otimes n}\rangle\bigr\rvert \leq \langle \gamma, |\hat{f}|\rangle^n<\infty.
\]
Since $\hat{f}_j\to \hat{f}$, $j\to\infty$, pointwise, the dominated convergence theorem implies that 
\begin{equation*}
  \lim_{j\to\infty }  \langle (\gamma)_n, \hat{f}_j^{\otimes n}\rangle=  \langle (\gamma)_n, \hat{f}^{\otimes n}\rangle, \qquad \gamma\in\Pi(\X).
\end{equation*}
Then, by the polarization identity, for any $f^\n\in\cF^\n(X)$,  we may also define $\langle (\gamma)_n, \hat{f}^\n\rangle$ for $\gamma\in\Pi(\X)$ and
\begin{equation*}
  \hat{f}^\n\bigl((s_1,x_1),\ldots,(s_n,x_n)\bigr) :=s_1\ldots s_n  f^\n(x_1,\ldots,x_n).
\end{equation*}
Then, by \eqref{fyd6e2},
\begin{equation}\label{eq:crucial}
  \langle P^\n(\eta),f^\n\rangle
  =\langle (\mathcal{R}^{-1}\eta)_n,\hat{f}^{(n)}\rangle, \qquad \eta\in\K(X).
\end{equation}

\begin{proposition}\label{prop:binom:cone}
  Let $f\in\cF(X)$. Then, for each $n\in\N$ and for each $\eta$, $\eta'\in\K(X)$ such that
    $\tau(\eta)\cap\tau(\eta')=\emptyset$,
    we have
    \begin{equation*}
    \langle P^{(n)}(\eta+\eta'),f^{\otimes n}\rangle
    =\sum_{k=0}^n\binom{n}{k}\langle P^{(k)}(\eta),f^{\otimes k}\rangle
    \langle P^{(n-k)}(\eta'),f^{\otimes (n-k)}\rangle.
    \end{equation*}
    \end{proposition}
  \begin{proof}
    By \eqref{def:cR}, the assumption $\tau(\eta)\cap\tau(\eta')=\emptyset$ implies that
    \begin{equation}
      \mathcal{R}^{-1}(\eta+\eta')=\mathcal{R}^{-1} \eta + \mathcal{R}^{-1}\eta'. \label{eqL3r}
    \end{equation}
    Then the statement follows immediately from \eqref{bin-id-marked} and \eqref{eq:crucial}.  
  \end{proof}

\begin{corollary}
Let $f\in \cF(X)$ and $n\in\N$. Then,
  
\noindent
1.~For each $\eta\in\K(X)$, $(s,x)\in\X$ such that $x\notin\tau(\eta)$,
\[
(D^+_{(s,x)}\langle P^{(n)}(\cdot),f^{\otimes n}\rangle)(\eta)=nsf(x)\langle P^{(n-1)}(\eta),f^{\otimes (n-1)}\rangle;
\]
2.~ For each $\eta\in\K(X)$ and $x\in\tau(\eta)$,
 \[
(D^-_x\langle P^{(n)}(\cdot),f^{\otimes n}\rangle)(\eta)=-ns_xf(x)\langle P^{(n-1)}(\eta-s_x\delta_x),f^{\otimes (n-1)}\rangle.
\] 
\end{corollary}

\begin{proof}
    1.~By \eqref{eqL3r}, we get from 
    \eqref{lowering} that
    \begin{align*}
      (D^+_{(s,x)}\langle P^{(n)}(\cdot),f^{\otimes n}\rangle)(\eta)&=
      \langle P^{(n)}(\eta+s\delta_x),f^{\otimes n}\rangle-\langle P^{(n)}(\eta),f^{\otimes n}\rangle\\
      &=\langle (\mathcal{R}^{-1}\eta+\delta_{(s,x)})_n,f^{\otimes n}\rangle-\langle (\mathcal{R}^{-1}\eta)_n,\hat{f}^{\otimes n}\rangle\\
      &=n\langle (\delta_{(s,x)}\odot (\mathcal{R}^{-1}\eta)_{n-1}),\hat{f}^{\otimes n}\rangle\\
      &= n\hat{f}(s,x)\langle (\mathcal{R}^{-1}\eta)_{n-1},\hat{f}^{\otimes (n-1)}\rangle\\
      &=nsf(x)\langle P^{(n-1)}(\eta ),f^{\otimes (n-1)}\rangle.
    \end{align*}
    
    2.~By item 1, 
 \begin{align*}
      (D^-_x\langle P^{(n)}(\cdot),f^{\otimes n}\rangle)(\eta)&=
      \langle P^{(n)}(\eta - s_x\delta x),f^{\otimes n}\rangle- \langle P^{(n)}(\eta),f^{\otimes n}\rangle
      \\& =-(D^+_{(s_x,x)}\langle P^{(n)}(\cdot),f^{\otimes n}\rangle)(\eta-s_x\delta_x)\\
      &=-ns_xf(x)\langle P^{(n-1)}(\eta -s_x\delta_x),f^{\otimes (n-1)}\rangle.\qedhere
    \end{align*}
\end{proof}

\begin{remark}
  We can also consider the generating function for $P^\n$, $n\geq0$. Namely, for each $f\in\cF(X)$, we have, by \eqref{eq:crucial}, \eqref{31Eq2} and \eqref{Kcohst},
  \begin{align*}
    \sum_{n=0}^\infty\frac{1}{n!}\langle P^\n(\eta),f^{\otimes n}\rangle
    &=\sum_{n=0}^\infty\frac{1}{n!}\langle(\mathcal{R}^{-1}\eta)_n,\hat{f}^{\otimes n}\rangle\\
    &=\exp\bigl(\langle \mathcal{R}^{-1}\eta,\log(1+\hat{f})\rangle\bigr)\\
    &=\exp\biggl(
\sum_{(s,x)\in\mathcal{R}^{-1}\eta}\log \bigl(1+sf(x)\bigr)  \biggr)\\
& = \exp\biggl(
  \sum_{x\in\tau(\eta)}\log \bigl(1+s_xf(x)\bigr)  \biggr)\\
  &= \prod_{x\in\tau(\eta)}\bigl(1+s_xf(x)\bigr)\\
  &=(Ke_\K(f))(\eta).
  \end{align*}
More generally,
  for any sequence $f^\n\in\cF^\n(X)$, $n\geq0$, one can define the function $F\in \cF_{\exp}(\K_0(X))$ given by, cf. \eqref{Eqr8},
  \[
    F(\xi):=s_{x_1}\ldots s_{x_n}f^\n(x_1,\ldots, x_n)
  \]
  for each $\xi=\sum_{i=1}^n s_{x_i}\delta_{x_i}\in\K_0(X)$, $n\geq1$; $F(0):=f^{(0)}\in\R$.
  Then
  \[
    \sum_{n=0}^\infty\frac{1}{n!}\langle P^\n(\eta),f^\n\rangle  = (KF)(\eta), \quad \eta\in\K(X).
  \]
\end{remark}

\subsection*{Acknowledgments}

MJO was supported by the Portuguese national funds through FCT--Funda{\c c}\~ao para a Ci\^encia e a Tecnologia, 
within the project UIDB/04561/2020 (https://doi.org/10.54499/UIDB/04561/2020).

%%%%%%%%%%%%%%%%%%%%%%%%%%%%%%%%%%%%%%%%%%%%%%%%%%%%%%%%%%%%%%%%%%%%%%%%%%%%%%%%

\end{document}